\documentclass[10pt]{article}
\usepackage{calc}
\usepackage[centertags]{amsmath}
\usepackage{latexsym}
\usepackage{amsfonts}
\usepackage{amssymb}
\usepackage{amsthm}
\usepackage{fancyhdr}
\usepackage [dvips]{epsfig}

\usepackage{newlfont}
\usepackage[latin1]{inputenc}
\usepackage[french,english]{babel}
\usepackage{graphicx}
\usepackage{t1enc}
\usepackage[matrix,arrow,frame,curve]{xy}
\theoremstyle{plain}
\newtheorem{lem}{Lemma}[section]
\newtheorem{prop}{Proposition}[section]
\newtheorem{theo}{Theorem}[section]
\newtheorem{definition}{Definition}[section]
\newtheorem{cor}{Corollary}[section]
\newtheorem{rem}{Remark}[section]

\newtheorem{ex}{Example}[section]
\theoremstyle{plain}

\title{On the Heisenberg invariance and the Elliptic Poisson tensors}

\author{ G. Ortenzi \footnote{
 Dipartimento di Matematica Pura e Applicazioni,
  Universit\`a degli Milano Bicocca
  Via R.Cozzi, 53
 20125, Milano, Italia
  E-mail address: giovanni.ortenzi@unimib.it},
  V. Rubtsov \footnote{
  Laboratoire Angevin de Recherche en Math\'ematiques
  Universit\'e D'Angers, D\'epartement de Math\'ematiques
  2, boulevard Lavoisier, 49045 Angers, France
  E-mail address: Volodya.Roubtsov@univ-angers.fr },
  S.R. Tagne Pelap\footnote{
  Laboratoire de Math\'ematiques, Universit\'e du Luxembourg
  E-mail address: serge.pelap@uni.lu}}
\begin{document}

\maketitle
 \begin{flushright}
 ITEP-TH-38/09
 \end{flushright}
\abstract{
\noindent
We study different algebraic and geometric properties of Heisenberg invariant Poisson
polynomial quadratic algebras. We show that these algebras are unimodular. The elliptic
Sklyanin-Odesskii-Feigin Poisson algebras $q_{n,k}(\mathcal E)$ are the main important example.
We classify all quadratic $H$-invariant Poisson tensors on ${\mathbb C}^n$ with $n\leq 6$ and show that
for $n\leq 5$ they coincide with the elliptic Sklyanin-Odesskii-Feigin Poisson algebras or with
their certain degenerations.

\section*{Introduction}
The elliptic or Sklyanin algebras are one of the most studied and important examples
of so-called \emph{algebras with quadratic relations}. This class of algebras possesses many
wonderful and useful properties. It is well-known (see, for example \cite{kiev} or \cite{ode3}) that
they can be defined as a quotient $T(V )/\langle R \rangle$ of the tensor algebra $T(V )$ of an
$n-$dimensional vector space $V$ over a space $R$ of quadratic relations. (In practice
$V$ is identified with sections of a degree $n$ vector bundle over a fixed elliptic curve
$\mathcal E$. It explains also the name "Elliptic algebras"). These class of Noetherian graded
associative algebras are Koszul, Cohen-Macaulay and have the Hilbert function
like a polynomial ring with $n$ variables.
The first $n = 4$ examples of such algebras were discovered by E. Sklyanin within his
studies of integrable discrete and continuous Landau-Lifshitz models by the Quantum
Inverse Scattering Method \cite{sky1,sky2}. These algebras are intensively studied \cite{odfe1,smsf,pelap1,pelap2}. Almost in the same time the elliptic algebras with
3 generators were discovered and studied by M. Artin, J. Tate, T. Stafford and their
students and collaborators among which we should mention M. Van den Bergh whose
income to studies of the Sklyanin elliptic algebras is very important \cite{ATvdB,arsc}.

The systematic studies of the Sklyanin elliptic algebras with any Gelfand-Kirillov
dimension $n\geq 3$ based upon some deformation quantization approach were proposed
by B. Feigin and A. Odesskii in their preprints and papers of 1980-90th. Here we will
quote some of them which are relevant to our aims (see the review paper \cite{ode3} and full
list of the references there).

In this paper we will focus on the important invariance property of the Sklyanin
elliptic algebras. Let us remind that if we have an $n-$dimensional vector space $V$ and
fixed a base $v_0,\ldots,v_{n-1}$ of $V$ then the \emph{Heisenberg group of level $n$ in the Schr\"{o}dinger
representation} is the subgroup $H_n\subset GL(V )$ generated by the operators
$$\sigma: v_i \to v_{i-1}; \tau : v_i \to \varepsilon_i(v_i); (\varepsilon_i)^n = 1; 0\leq i \leq n-1.$$
This group has order $n^3$ and is a central extension
$$1 \to {\mathbb U}_n \to H_n \to \mathbb Z_n \times \mathbb Z_n \to 1,$$
where $\mathbb U_n$ is the group of $n-$th roots of unity. \\
The property to be Heisenberg invariant is quite often in quantum and classical integrable systems.
Probably the first important manifestation of it was detected by A. Belavin \cite{Bel}. He has shown that the
hidden symmetry of 1 dimensional integrable models (which is usually expressed by some Yang-Baxter equation) is
contained in the Heisenberg invariance. \\
The Heisenberg group action provides the automorphisms
of the Sklyanin algebra which are compatible with the grading and defines also an
action on the "quasi-classical" limit of the Sklyanin algebras - the elliptic quadratic
Poisson structures on $\mathbb P^{n-1}$. They are identified with Poisson structures on some moduli
spaces of the degree $n$ and rank $k + 1$ vector bundles with parabolic structure (=
the flag $0\subset F\subset {\mathbb C}^{k+1}$ on the elliptic curve $\mathcal E$). We will denote these elliptic Poisson algebras by $q_{n,k}(\mathcal E)$.
The algebras $q_{n,k}(\mathcal E)$ arise in the Feigin-Odesskii "deformational"
approach and
form a subclass of \emph{polynomial Poisson structures}. We want to show in this paper
that the extension of the Heisenberg group action to more wide family of polynomial
Poisson structures guarantees also some good and useful features and among
them:  the \emph{unimodularity} property and the triviality of the modular class. This notion
was introduced by A. Weinstein (who studied it with his collaborators in \cite{ELW})
and, independently in the holomorphic context, by Brylinsky and Zuckerman \cite{BrZuk}.
The modular class is the element of the 1-st Poisson cohomology group $HP^1(M)$ of
a Poisson manifold $M$. The "Basic Theorem of Mechanics" says that this group is
isomorphic to the quotient $Can(M)/Ham(M)$, where $Can(M)\subset \Gamma(TM)$ (resp.
$Ham(M)\subseteq Can(M))$ denote the space of Poisson (resp. Hamiltonian) vector fields
on $M$. Recently, the importance of the unimodularity property was recognized by
Dolgushev \cite{dol}. He has shown that for an affine Poisson variety with trivial canonical
bundle the unimodularity of the Poisson structure is equivalent to the fact that the
deformational quantization $A_h$ of the function algebra $A$ on this variety is isomorphic
to the \emph{dualizing Van den Bergh module} of $A_h$ entering in the Van den Bergh duality
between Hochschild homology and Hochschild cohomology of $A_h$ (the isomorphism is in the sense of bimodules).
One of the
aims of the paper is to continue of studies of the polynomial Poisson algebras which
were started in \cite{odru} and, in particular, to give some classification of low-dimensional
Heisenberg-invariant \emph{quadratic} polynomial Poisson tensors. Our results complete the
computations of A. Odesskii for $n = 5$ (unpublished) and shed some light on the
geometry behind the elliptic Poisson algebras $q_{5,1}(\mathcal E)$ and $q_{5,2}(\mathcal E)$.
These algebras are $\eta\to 0$ limits of the Sklyanin algebras $Q_{5,1}(\mathcal E; \eta)$ and $Q_{5,2}(\mathcal E;\eta)$, which were described in the famous "Kiev preprint" of B. Feigin and A. Odesskii \cite{kiev}.
They have shown in particular that the algebra $Q_{5,1}(\mathcal E; \eta)$ embeds as a subalgebra into
the algebra $Q_{5,2}(\mathcal E; \eta)$ and vice versa. We will show in the subsequent paper \cite{ORT2} that these embedding homomorphisms
are the "quantum analogs" of the quadro-cubic Cremona transformations
of $\mathbb P^4$ \cite{SemRoth} and will check that they are Poisson birational morphisms which presumably
correspond to compositions of Polishchuk birational mappings between moduli spaces
of triples $Mod(E_1;E_2; f)$ of vector bundles and their morphisms on an elliptic curve
\cite{polys}.

Our paper is written with double aims: to present some easy but maybe unknown (to the best of our knowledge) fact about polynomial Heisenberg invariant Poisson structures e.g. unimodularity,  and, on the other hand, to give a classification of quadratic structures in small dimensions to complete the results of \cite{odru}.

The paper organizes as follows: The Section 1 covers some known facts about  polynomial quadratic and elliptic (Poisson) algebras. In Section 2 we remind a definition of the Heisenberg group in the Schroedinger representation and describe its action on Poisson polynomial tensors. The unimodularity of Heisenberg invariant quadratic Poisson
structures is a subject of the Section 3. We close this section with a short observation about the Calabi-Yau property
of the deformation quantization of Heisenberg-invariant Poisson polynomial structures.
Section 4 is devoted to a classification of Heisenberg-invariant quadratic Poisson structures up to $d=6$. We obtain in this way three type of solutions of the classification problem for $d=5$ and we will show in the subsequent paper \cite{ORT2} that these solutions are exhausted by two Sklyanin-Odesskii-Feigin elliptic algebras plus skew-polynomials (among them the $12$ degenerations of the Sklyanin-Odesskii-Feigin algebras which correspond to \emph{pentagonal} degenerations of the underlying family of elliptic curves in "cuspidal points").\
Conclusions and some possible directions of future work are presented in Section 5.

\section{Preliminary facts}
In this section, $K$ is a field of characteristic zero.
\subsection{Poisson algebras and Poisson manifold}
 Let $\mathcal R$ be a commutative $K$-algebra. $\mathcal R$ is a Poisson algebra if $\mathcal R$ is a Lie algebra such that the bracket is a biderivation. In other words, $\mathcal R$ is a Poisson algebra if $\mathcal R$ is endowed with a $K$-bilinear map $\{\cdot, \cdot\} : \mathcal R\times \mathcal R\longrightarrow\mathcal R$ which satisfies the following properties:\\
$\begin{array}{lr}
   \{a, bc\}=\{a, b\}c+b\{a, c\} & \mbox{(The Leibniz rule)};\\
   \{a, b\}=-\{b, a\}& \mbox{(antisymmetry)};\\
  \{a, \{b, c\}\}+\{b, \{c, a\}\}+\{c, \{a, b\}\}=0& \mbox{(The Jacobi identity)}.\\
\end{array}$\\
where $a, b, c\in\mathcal R.$\\
One can also say that $\mathcal R$ is endowed with a Poisson structure and that the bracket $\{\cdot, \cdot\}$ is called the Poisson bracket on $\mathcal R.$
An element $a\in\mathcal R$ such that $\{a,b\}=0$ for all $b\in\mathcal R$ is called a Casimir.\\
A manifold $M$ (smooth, algebraic,...) is said to be a Poisson manifold if its functions algebra $\mathcal A$ ($C^\infty(M)$, regular,...) is endowed with a Poisson bracket.\\
Let us consider some specific examples of Poisson algebras interesting for our aims. Let
$$q_1=\frac{1}{2}(x_0^2 + x_2^2)+kx_1x_3,$$
$$q_2=\frac{1}{2}(x_1^2 + x_3^2)+kx_0x_2,$$
be two elements of $\mathbb C[x_0, x_1, x_2, x_3]$ where $k\in\mathbb C.$\\
We have a Poisson structure $\pi$ on $\mathbb C^4$ or, more generally, on $\mathbb C[x_0, x_1, x_2, x_3]$  by the formula:
$$\{f, g\}_{\pi}:=\frac{df\wedge dg\wedge dq_1\wedge dq_2}{dx_0\wedge dx_1\wedge dx_2\wedge dx_3}.$$
Then the brackets between the coordinate functions are defined by (mod $4$):
\begin{eqnarray*}
\{x_i, x_{i+1}\}=k^2x_ix_{i+1}-x_{i+2}x_{i+3}, \\
\{x_i, x_{i+2}\}=k(x_{i+3}^2-x_{i+1}^2),\ i=1, 2, 3, 4.
\end{eqnarray*}
This algebra will be denoted by $q_4(\mathcal E)$ and called the Sklyanin Poisson algebra, where $\mathcal E$ represents an elliptic curve, which parameterizes the algebra (via $k$).\\
 We can also think of this curve $\mathcal E$ as a geometric interpretation of the couple $q_1=0$, $q_2=0$ embedded in $\mathcal CP^3$ (as was observed in Sklyanin's initial paper).\\
A possible generalization can be obtained considering $n-2$ polynomials $Q_i$ in $K^n$ with coordinates $x_i$, $i=0,...,n-1,$ we can define, for any polynomial $\lambda\in K[x_0,...,x_{n-1}]$, a bilinear differential operation :
$$\{\cdot ,\cdot\} : K[x_0,...,x_{n-1}]\otimes K[x_0,...,x_{n-1}]\longrightarrow K[x_0,...,x_{n-1}]$$
by the formula
\begin{equation}\label{q}
\{f,g\}=\lambda\frac{df\wedge dg\wedge dQ_1\wedge...\wedge dQ_{n-2}}{dx_0\wedge dx_1\wedge...\wedge dx_{n-1}},\quad f,g\in K \left[ x_0,...,x_{n-1} \right].
 \end{equation}
This operation gives a Poisson algebra structure on $K[x_0,...,x_{n-1}].$
The polynomials $Q_i, i=1,...,n-2$ are Casimir functions for the brackets (\ref{q}) and any Poisson structure on $K^n$, with $n-2$ generic Casimirs $Q_i$, is written in this form. Every Poisson structure of this form is called a Jacobian Poisson structure (JPS) \cite{khi1,khi2} for $\lambda=1$.\\
The case $n=4$ in $(\ref{q})$ corresponds to the classical generalized Sklyanin quadratic Poisson algebra. In the next subsection we present a different generalization to every dimension of the Sklyanin algebras which is the starting point of our paper.
\subsection{Elliptic Poisson algebras $q_{n}(\mathcal E, \eta)$ and $q_{n, k}(\mathcal E, \eta)$}
 These algebras, defined by Odesskii and Feigin, arise as quasi-classical limits of elliptic associative algebras $Q_n(\mathcal E, \eta)$ and $Q_{n, k}(\mathcal E, \eta)$ \cite{odfe1, ode1}.\\
Let $\Gamma=\mathbb Z+\tau\mathbb Z\subset\mathbb C,$ be an integral lattice generated by $1$ and $\tau\in\mathbb C,$ with $\textnormal{Im}\tau > 0.$
Consider the elliptic curve $\mathcal E=\mathbb C/\Gamma$ and a point $\eta$ on this curve.\\
In their article \cite{ode1}, given $k< n,$ mutually prime, Odesskii and Feigin construct an algebra, called  elliptic, $Q_{n, k}(\mathcal E, \eta),$ as an algebra defined by $n$ generators  $\{x_i, i\in\mathbb Z/n\mathbb Z\}$ and the following relations
 \begin{equation}
  \displaystyle\sum_{r\in\mathbb{Z}/n\mathbb{Z}}\frac{\theta_{j-i+r(k-1)}(0)}
 {\theta_{kr}(\eta)\theta_{j-i-r}(-\eta)}x_{j-r}x_{i+r}=0, \ \ i\neq j, i, j\in\mathbb Z/n\mathbb Z
 \end{equation}
 where $\theta_{\alpha}$ are theta functions \cite{ode1}.
\\
 We have the following properties :
 \begin{enumerate}
   \item The center of the algebra $Q_{n, k}(\mathcal E, \eta),$ for generic $\mathcal E$ and $\eta,$ is the algebra of polynomial of $m=pgcd(n, k+1)$ variables of degree $n/m;$
   \item $Q_{n, k}(\mathcal E, 0)=\mathbb C[x_1, \cdots, x_n]$ is commutative;
   \item $Q_{n, n-1}(\mathcal E, \eta)=\mathbb C[x_1, \cdots, x_n]$ is commutative for all $\eta$;
   \item $Q_{n, k}(\mathcal E, \eta)\simeq Q_{n, k'}(\mathcal E, \eta),$ if $kk'\equiv 1$ (mod $n$);
   \item the maps $x_i\mapsto x_{i+1}$ et $x_i\mapsto \varepsilon^ix_i$, where $\varepsilon^n=1,$ define  automorphisms of the algebra $Q_{n, k}(\mathcal E, \eta)$;
   \item the algebras $Q_{n,k}(\mathcal E, \eta)$ are deformations of polynomial algebras. The associated Poisson structure is denoted by $q_{n,k}(\mathcal E, \eta)$;
   \item among the algebras $q_{n,k}(\mathcal E, \eta),$ only $q_{3}(\mathcal E, \eta)$ and the Sklyanin algebra $q_{4}(\mathcal E, \eta)$ are Jacobian Poisson structures.
 \end{enumerate}
The explicit form of $q_{n,k}$ is not known in generic dimension.
\section{The H-invariance}
In the previous section we have introduced one of the most important classes of Poisson algebras related to an elliptic curve. We  want now to change the point of view and to concentrate our study on a particular geometric property of the Poisson bivector related to $q_{n,k}$. As we have reminded in the previous section the discrete Heisenberg group naturally acting on the elliptic curve $\mathcal{E}$ induces a class of automorphisms on the elliptic algebra associated which leaves invariant the algebra structure.\\
Let $V$ be a complex vector space of dimension $n$ and $e_0,\cdots, e_{n-1}$ a basis of $V$.\\
Let us consider $\sigma, \tau$ of $GL(V)$ defined by:
$$\sigma(e_m)=e_{m-1},$$
$$\tau(e_m)=\varepsilon^me_m$$
where $\varepsilon=e^{\frac{2\pi i}{n}}$.
The subspace $H_n\subset GL(V)$ generated by $\sigma$ and $\tau$ is called the Heisenberg group of dimension $n.$\par
From now on we suppose that $V=\mathbb C^n,$ with $x_0, x_1,\cdots, x_{n-1}$ as basis and we consider its algebras of regular functions $\mathbb C[x_0,x_1,\cdots, x_{n-1}].$\\
$\sigma$ and $\tau$ act by automorphisms on the algebra $\mathbb C[x_0,x_1,\cdots, x_{n-1}]$ as follows :
$$\sigma\cdot(\alpha x_0^{\alpha_0}x_1^{\alpha_1}\cdots x_{n-1}^{\alpha_{n-1}})=\alpha x_0^{\alpha_{n-1}}x_1^{\alpha_0}\cdots x_{n-1}^{\alpha_{n-2}};$$
$$\tau\cdot(\alpha x_0^{\alpha_0}x_1^{\alpha_1}\cdots x_{n-1}^{\alpha_{n-1}})=\varepsilon^{\alpha_1+2\alpha_2+\cdots+(n-1)\alpha_{n-1}}\alpha x_0^{\alpha_0}x_1^{\alpha_1}\cdots x_{n-1}^{\alpha_{n-1}}.$$
Then the $\tau$ degree of a monomial $M=\alpha x_0^{\alpha_0}x_1^{\alpha_1}\cdots x_{n-1}^{\alpha_{n-1}}$ is by definition $\alpha_1+2\alpha_2+\cdots+(n-1)\alpha_{n-1}\in\mathbb Z/n\mathbb Z$.
The $\tau$ degree of $M$ is denoted $\tau$-$deg(M)$. A $\tau$ degree of a polynomial is the highest $\tau$-degree of its monomials.
\begin{definition}
A Poisson bracket $\{\cdot, \cdot\}$ on $\mathcal R=\mathbb C[x_0,x_1,\cdots, x_{n-1}]$ is said to be $H$-invariant for the variables $x_0,\cdots, x_{n-1}$ if the automorphisms $\sigma$ and $\tau$ are Poisson morphisms.\\
In other words, the following diagrams commute:
$$\xymatrix{\mathcal R\times\mathcal R \displaystyle{\ar[r]^{\sigma\times \sigma}}\displaystyle{\ar[d]_{\{\cdot, \cdot\}}}&
\mathcal R\times\mathcal R\displaystyle{\ar[d]^{\{\cdot, \cdot\}}}\\
\mathcal R \displaystyle{\ar[r]^{\sigma}}&\mathcal R
}$$
$$\xymatrix{\mathcal R\times\mathcal R \displaystyle{\ar[r]^{\tau\times \tau}}\displaystyle{\ar[d]_{\{\cdot, \cdot\}}}&
\mathcal R\times\mathcal R\displaystyle{\ar[d]^{\{\cdot, \cdot\}}}\\
\mathcal R \displaystyle{\ar[r]^{\tau}}&\mathcal R
}.$$
i.e., $\sigma\cdot\{x_i, x_j\}=\{\sigma\cdot x_i, \sigma\cdot x_j\}$ and $\tau\cdot\{x_i, x_j\}=\{\tau\cdot x_i, \tau\cdot x_j\},$ for all $i, j.$
\end{definition}
A matrix $P_{ij}$ which defines an $H$-invariant Poisson tensor must satisfy four properties: antisymmetry, Jacobi identities, and $\sigma$ and $\tau$ invariance.
\begin{itemize}
 \item $\sigma$ invariance \\
  The condition
\begin{equation}
\{ \sigma x_i,\sigma x_j \} = P_{ij}(\sigma x_0, \sigma x_1 \dots, \sigma x_{n-2}, \sigma x_{n-1})
\end{equation}
implies
\begin{equation}
  P_{i+1\ j+1}( x_0, x_1 \dots, x_{n-2}, x_{n-1} ) = P_{ij}( x_1, x_2 \dots, x_{n-1}, x_0 ).
\end{equation}
\item $\tau$ invariance \\
  The condition
 \begin{equation}
\{ \tau x_i,\tau x_j \} = P_{ij}(\tau x_0, \tau x_1 \dots, \tau  x_{n-2}, \tau  x_{n-1})
\end{equation}
implies
\begin{equation}
  \epsilon^{i+j} P_{ij}( x_0, x_1 \dots, x_{n-2}, x_{n-1} ) = P_{ij}( x_0, \epsilon x_1 \dots, \epsilon^{n-1} x_{n-1}).
\end{equation}
The $\tau$ invariance is, in some sense, a ``discrete'' homogeneity and its intersection with the $\sigma$ invariance is small. If we restrict to polynomial Poisson tensors, then the previous two conditions can be satisfied only by sum of monomials of order $2+sn, \ s \geq 0$. Actually the  structure of the element $P_{ij}$ is given by $a_{i,j}\delta_{n-i-j, 0}+ b x_{i+j} + \sum_k c_k x_k x_{i+j-k} + \sum_{k,l} d_k x_k x_l x_{i+j-k-l} + \dots$ by the $\tau$  invariance. The $\sigma$ invariance imposes that $P_{i+1\ j+1}=a_{i,j}\delta_{n-i-j, 0}+ b x_{i+j+1} + \sum_k c_k x_{k+1} x_{i+j-k+1} + \sum_{k,l} d_k x_{k+1} x_{l+1} x_{i+j-k-l+1} + \dots,$ and thus $a_{i,j}=0$ for all $i, j.$
\item antisymmetry \\
From the $\sigma$ invariance it is obvious that the matrix associated with the Poisson tensor has only n-1 independent entries modulo cyclic permutations of the co-ordinates. However the antisymmetry imposes $P_{n0}(x_0,\dots, x_{n-1})=-P_{0n}(x_0,\dots, x_{n-1})$ and the $\sigma$ invariance $P_{n0}(x_0,x_1,\dots, x_{n-1})=P_{01}(x_{n-1},x_0,\dots, x_{n-2})$. The intersection of the two conditions reduce the number of free functions to $\left[ n/2 \right] $. In the case of even dimensions, the antisymmetry imposes a further restriction on the function $P_{0\ n/2}=-P_{n/2\ 0}$ which has to satisfy $P_{ n/2 \ 0}(x_0, \dots, x_{n-1})=
P_{ 0 \ n/2}(x_{n/2}, \dots, x_{3n/2-1})$ by $\sigma$ invariance. \\
Resuming the properties of the $H$-invariant Poisson tensor we obtain that they have the following form:
\begin{eqnarray*}\label{H-P}
\left(
\begin{array}{ccccccccc}
0&P^0_1&\dots&\dots&P^0_{\frac{n-1}{2}}&-P^{\frac{n+1}{2}}_{\frac{n-1}{2}}&\dots&\dots&-P^{n-1}_1\\
-P^0_1 & 0 & P^1_1&\dots&\dots&P^1_{\frac{n-1}{2}}& \ddots&&\vdots   \\
\vdots&-P^1_1& 0 & P^2_1&&&\ddots&\ddots&\vdots\\
\vdots&&- P^2_1&0&\ddots&&&\ddots&-P^{n-1}_{\frac{n-1}{2}} \\
-P^0_{\frac{n-1}{2}}&&&\ddots&\ddots&\ddots&&&P^\frac{n-1}{2}_{\frac{n-1}{2}} \\
P^\frac{n+1}{2}_\frac{n-1}{2}&-P^1_{\frac{n-1}{2}}   &&&\ddots&0&\ddots&&\vdots\\
\vdots&\ddots&\ddots&&&\ddots&0&P^{n-3}_1&\vdots\\
\vdots&&\ddots&\ddots&&&-P^{n-3}_1&0&P^{n-2}_1 \\
P^{n-1}_1&\dots&\dots&P^{n-1}_{\frac{n-1}{2}}&-P^{\frac{n-1}{2}}_{\frac{n-1}{2}}&\dots&\dots&-P^{n-1}_1&0
\end{array}
\right) &&for\ n\ odd
\end{eqnarray*}
and
\begin{eqnarray*}
\left(
\begin{array}{ccccccccc}
0&P^0_1&\dots&\dots&P^0_{\frac{n-2}{2}}&P^{0}_{\frac{n}{2}}&\dots&\dots&-P^{n-1}_1\\
-P^0_1 & 0 & P^1_1&\dots&\dots&P^1_{\frac{n-2}{2}}& \ddots&&\vdots   \\
\vdots&-P^1_1& 0 & P^2_1&&&\ddots&\ddots&\vdots\\
\vdots&&- P^2_1&0&\ddots&&&\ddots&P^{\frac{n-2}{2}}_{\frac{n}{2}} \\
-P^0_{\frac{n-2}{2}}&&&\ddots&\ddots&\ddots&&&P^\frac{n}{2}_{\frac{n-2}{2}} \\
-P^0_\frac{n}{2}&-P^1_{\frac{n-1}{2}}   &&&\ddots&0&\ddots&&\vdots\\
\vdots&\ddots&\ddots&&&\ddots&0&P^{n-3}_1&\vdots\\
\vdots&&\ddots&\ddots&&&-P^{n-3}_1&0&P^{n-2}_1 \\
P^{n-1}_1&\dots&\dots&-P^{\frac{n-2}{2}}_{\frac{n}{2}}&-P^{\frac{n}{2}}_{\frac{n-2}{2}}&\dots&\dots&-P^{n-2}_1&0
\end{array}
\right) &&for\ n\  even,
\end{eqnarray*}
where $P_i^0=P_i(x_0,x_1,\dots,x_{n-1})$, $P_i^k=P_i(x_k,x_{1+k},\dots,x_{{n-1+k}})$  and $P_i(x_k)$ are  functions such that $P_i(\epsilon^k x_k)=\epsilon^i P_{i}(x_k)$.
\item Jacobi relations \\ The complete study of the Jacobi relations could allow us a complete classification of $H$-invariant Poisson tensors. The task is long in general and we concentrate our study on the quadratic Poisson tensors. The (results of the) computations are the subject of the Section $4$, for small dimensions.
\end{itemize}
\begin{ex}
The Sklyanin-Odesskii-Feigin Poisson algebras $q_{n,k}(\mathcal E)$ are $H$-invariant Poisson algebras.
\end{ex}
\begin{prop}
If $\{\cdot, \cdot\}$ is an $H$-invariant polynomial Poisson bracket, the usual polynomial degree of the monomial of $\{x_i, x_j\}$ has the form $2+sn,\ s \in\mathbb N$.
\end{prop}
\begin{proof}
Since the $H$-action does not change the usual degree of a polynomial, we will do the proof just for a homogeneous case. Let $\{\cdot, \cdot\}$ be a non trivial $H$-invariant homogeneous Poisson of usual polynomial degree $N.$ Then $\{x_i, x_j\}$ have the form:
\begin{equation}
\{x_i, x_j\}=\displaystyle{\sum_{k_{1},\cdots, k_{N-1}}} a_{k_{1},\cdots, k_{N-1}}x_{k_1}\cdots x_{k_{N-1}}x_{i+j-k_{1}-\cdots-k_{N-1}}.
\end{equation}
We have:
\begin{equation}
\sigma\cdot\{x_i, x_j\}=\displaystyle{\sum_{k_{1},\cdots, k_{N-1}}} a_{k_{1},\cdots, k_{N-1}}x_{k_1+1}\cdots x_{k_{N-1}+1}x_{i+j-k_{1}-\cdots-k_{N-1}+1}.
\end{equation}
Since $\sigma\cdot\{x_i, x_j\}=\{x_{i+1}, x_{j+1}\}$, $\tau$-$deg(\{x_{i+1}, x_{j+1}\})=i+j+2=i+j+N.$ Hence $N=2+sn,$ $s\in\mathbb N.$
\end{proof}
\begin{rem}
We want to stress that there are non quadratic Heisenberg invariant polynomial Poisson structures. The simplest example is given by $H_3$ invariant Jacobi Poisson structure given by the Casimir
\begin{equation}
 C=\frac{1}{2}{x_0}^2{x_1}^2{x_2}^2.
\end{equation}
\end{rem}
\section{The unimodularity and  H-invariant Poisson tensors}
\subsection{Unimodular Poisson structures}
Let $\mathcal R=K[x_0, \cdots, x_{n-1}]$ be the polynomial algebra with $n$ variables.\\
 We recall that the $\mathcal R$-module of K\"ahler differentials of $\mathcal R$ is denoted by $\Omega^1(\mathcal R)$ and the graded $\mathcal R$-module $\Omega^p(\mathcal R) :=\bigwedge^p\Omega^1(\mathcal R)$ is the module of all K\"ahler $p$-differential.
As a vector space (resp. as an $\mathcal R$-module) $\Omega^p(\mathcal R)$ is generated by elements of the form $FdF_1\wedge...\wedge dF_p$ (resp. of the form $dF_1\wedge...\wedge dF_p$) where $F, F_i\in\mathcal R$, $i=1,...,p$.
We denote by $\Omega^{\bullet}(\mathcal R)=\displaystyle{\oplus_{p\in\mathbb N}}\Omega^p(\mathcal R)$, with the convention that $\Omega^0(\mathcal R)=\mathcal R$, the space of all K\"ahler differentials.\\
The differential $d : \mathcal R\longrightarrow\Omega^1(\mathcal R)$ extends to a graded $K$-linear map $$d :\Omega^\bullet(\mathcal R)\longrightarrow\Omega^{\bullet+1}(\mathcal R)$$
by setting :
$$
d(GdF_1\wedge...\wedge dF_p) :=dG\wedge dF_1\wedge...\wedge dF_p
$$
for $G,F_1,...,F_p\in\mathcal R$, where $p\in\mathbb{N}$. It is called the de Rham  differential. It is a graded derivation, of degree $1$, of $(\Omega^\bullet(\mathcal R), \wedge)$, such that $d^2=0$.\\
A skew-symmetric $k$-linear map $P\in\mbox{Hom}_K(\wedge^k\mathcal R,\mathcal R)$ is called a skew-symmetric $k$-derivation of $\mathcal R$ with values in $\mathcal R$ if it is a derivation in each of its arguments.\\
The $\mathcal R$-module of skew-symmetric $k$-derivation is denoted by $\mathcal X^k(\mathcal R).$ We define the graded $\mathcal R$-module
$$\mathcal X^\bullet(\mathcal R):=\displaystyle{\bigoplus_{k\in\mathbb N}}\mathcal X^k(\mathcal R)$$
whose elements are called skew-symmetric multi-derivations. By convention, the first term in this sum, $\mathcal X^0$, is $\mathcal R$, and $\mathcal X^k(\mathcal R):=\{0\}$ for $k<0$.\\
One can observe that if $\{\cdot,\cdot\}$ is a Poisson structure on algebra $\mathcal R,$ then $\{\cdot, \cdot\}\in\mathcal X^2(\mathcal R).$\\
We consider now the affine space of dimension $n,$ $K^n$ and its algebra of regular functions $\mathcal R=K[x_0,...,x_{n-1}]$.\\
We denote by $S_{p,q}$ the set of all $(p,q)-$shuffles, that is permutations $\sigma$ of the set $\{0,...,p+q-1\}$ such that $\sigma(0)<...<\sigma(p-1)$ and $\sigma(p)<...<\sigma(p+q-1), p, q\in\mathbb N $.\\
The family of maps $\star : \mathcal X^k(\mathcal A)\longrightarrow\Omega^{n-k}(\mathcal A)$, defined by :
$$\star Q=\displaystyle{\sum_{\alpha\in S_{k,n-k}}}\epsilon(\alpha)Q(x_{\alpha(0)},...,x_{\alpha(k-1)})dx_{\alpha(k)}\wedge...\wedge dx_{\alpha(n-1)}$$
are isomorphisms.\\
Let $D_\bullet$ be the map : $D_\bullet := \star^{-1}\circ d\circ\star : \mathcal X^\bullet(\mathcal R)\longrightarrow\mathcal X^{\bullet -1}(\mathcal R)$, where $d$ is the de Rham differential.\\
Let us remain that the Poisson coboundary operator associated with a Poisson algebra $(\mathcal A, \{\cdot,\cdot\})$ and denoted by $\delta : \mathcal X^\bullet(\mathcal A)\longrightarrow\mathcal X^{\bullet+1}(\mathcal A)$, is given by:
$$\delta^k(Q)(F_0, F_1,..., F_k)=\displaystyle{\sum_{0\leq i\leq k}}(-1)^{i}\{F_i, Q(F_0, F_1,...,\widehat F_i,..., F_k)\}$$
$$+\displaystyle{\sum_{0\leq i<j\leq k}} (-1)^{i+j}Q(\{F_i, F_j\}, F_0, ...,\widehat F_i,..., \widehat F_j,...,F_k)$$\\
where $F_0,...,F_k\in\mathcal A,$ and $Q\in\mathcal X^k(\mathcal A).$\\
One can check, by a direct computation, that $\delta^k$ is well-defined and that it is a coboundary o\-pe\-ra\-tor, $\delta^{k+1}\circ\delta^{k}=0.$\\
The cohomology of this complex is called the Poisson cohomology associated with $(\mathcal A, \pi)$ and denoted by $PH^{\bullet}(\mathcal A, \pi).$\\
Let $\pi$ a Poisson structure on $\mathbb C[x_0,\cdots, x_{n-1}].$
\begin{equation}
\pi=\displaystyle{\sum_{0\leq i, j\leq n-1}}\{x_i, x_j\}\frac{\partial}{\partial x_i}\wedge\frac{\partial}{\partial x_j}.
\end{equation}
Then
$$\begin{array}{ll}
    \star(\pi)&=\displaystyle{\sum_{0\leq i, j\leq n-1}}\displaystyle{\sum_{\alpha\in S_{2, n-2}}}(-1)^{|\alpha|}\delta_{i, \alpha(i)}\delta_{j, \alpha(j)}\{x_i, x_j\}dx_{\alpha(2)}\wedge...\wedge dx_{\alpha(n-1)}\\ &\\
    &=\displaystyle{\sum_{0\leq i, j\leq n-1}}(-1)^{i+j-1}\{x_i, x_j\}dx_0\wedge...\wedge \widehat{dx_i}\wedge...\wedge \widehat{dx_j}\wedge...\wedge {dx_{n-1}}.
\end{array}$$
Therefore
$$\begin{array}{ll}
d(\star(\pi))&=\displaystyle{\sum_{0\leq i, j\leq n-1}}\displaystyle{\sum_{l=0}^{n-1}}(-1)^{i+j-1}\frac{\partial\{x_i, x_j\}}{\partial x_l}dx_l\wedge dx_0\wedge...\wedge \widehat{dx_i}\wedge...\wedge \widehat{dx_j}\wedge...\wedge {dx_{n-1}}\\ & \\
   &=\displaystyle{\sum_{0\leq i, j\leq n-1}}(-1)^{i}\frac{\partial\{x_i, x_j\}}{\partial x_j}dx_0\wedge...\wedge \widehat{dx_i}\wedge...\wedge {dx_{n-1}}-\\
   &\  - \displaystyle{\sum_{0\leq i, j\leq n-1}}(-1)^{j}\frac{\partial\{x_i, x_j\}}{\partial x_i}dx_0\wedge...\wedge \widehat{dx_j}\wedge...\wedge {dx_{n-1}}.
\end{array}$$
Then
$$\begin{array}{ll}
D_2(\pi)&=\displaystyle{\sum_{0\leq i, j\leq n-1}}\frac{\partial\{x_i, x_j\}}{\partial x_j}\frac{\partial}{\partial x_i}-\displaystyle{\sum_{0\leq i, j\leq n-1}}\frac{\partial\{x_i, x_j\}}{\partial x_i}\frac{\partial}{\partial x_j}\\
&=2\displaystyle{\sum_{0\leq i, j\leq n-1}}\frac{\partial\{x_i, x_j\}}{\partial x_j}\frac{\partial}{\partial x_i}.
\end{array}$$
\begin{prop}\cite{xu1}
Let $\pi$ be a Poisson tensor. $D_2(\pi)$ is a Poisson cocycle: $\delta^1(D_2(\pi))=0$.
\end{prop}

The modular class of a Poisson tensor $\pi$ on $\mathcal R$ is the class of $D_2(\pi)$ in the first Poisson cohomological group $PH^1(\mathcal R, \pi),$ associated to $\pi.$ A Poisson tensor is said to be unimodular if its modular class is trivial \cite{xu1}.

From \cite{khi1} and \cite{khi2} we know that all JPS are unimodular. Therefore $H$-invariant Poisson structure in dimension 3 and 4 are unimodular. What happens in higher dimension?

\subsection{The unimodularity and  H-invariant Poisson structures}
Let $\pi$ a Poisson structure on $\mathbb C[x_0,\cdots, x_{n-1}].$
\begin{equation}
\pi=\displaystyle{\sum_{0\leq i, j\leq n-1}}\{x_i, x_j\}\frac{\partial}{\partial x_i}\wedge\frac{\partial}{\partial x_j}.
\end{equation}
\begin{prop}
A Poisson structure $\pi=\{\cdot, \cdot\}$ is  unimodular if $$\displaystyle{\sum_{i=0}^{n-1}}\frac{\partial\{x_k, x_i\}}{\partial x_i}=0,$$ for all $k=0\cdots, n-1.$
\end{prop}
\begin{lem}\label{Hder}
Let $P\in \mathcal A=K[x_0,\cdots, x_{n-1}].$ \\
For all $i\in\{0, \cdots, n-1\},$ $$\sigma\cdot\frac{\partial F}{\partial x_i}= \frac{\partial(\sigma\cdot F)}{\partial(\sigma\cdot x_i)}.$$
\end{lem}
\begin{proof}
Consider the monomial $M=\alpha x_0^{\alpha_0}x_1^{\alpha_1}\cdots x_{i-1}^{\alpha_{i-1}}x_i^{\alpha_i}x_{i+1}^{\alpha_{i+1}}\cdots x_{n-1}^{\alpha_{n-1}}.$ We have:
$$\sigma\cdot M=\alpha x_{0}^{\alpha_{n-1}}x_1^{\alpha_0}\cdots x_{i-1}^{\alpha_{i-2}}x_i^{\alpha_{i-1}}x_{i+1}^{\alpha_{i}}\cdots x_{n-1}^{\alpha_{n-2}}.$$
Then:
$$\begin{array}{ll}
\frac{\partial (\sigma\cdot M)}{\partial x_{i+1}}&=\alpha\alpha_i x_{0}^{\alpha_{n-1}}x_1^{\alpha_0}\cdots x_{i-1}^{\alpha_{i-2}}x_i^{\alpha_{i-1}}x_{i+1}^{\alpha_{i-1}}\cdots x_{n-1}^{\alpha_{n-2}}\\
&=\sigma\cdot(\frac{\partial M}{\partial x_{i}})
\end{array}
$$
\end{proof}
\begin{prop}\label{hqc}
A $H$-invariant Poisson structure $\pi=\{\cdot, \cdot\}$ is  unimodular if  $$\displaystyle{\sum_{i=0}^{n-1}}\frac{\partial\{x_0, x_i\}}{\partial x_i}=0.$$
\end{prop}
\begin{proof}
This is direct consequence of the lemma (\ref{Hder}).
\end{proof}

\begin{theo}
Every $H$-invariant quadratic Poisson structure is unimodular.
\end{theo}
\begin{proof}
We will give the proof in the odd case. The proof for the even case is done in a similar way with the remark that $\frac{\partial P_{0,\frac{n}{2}}}{\partial x_{\frac{n}{2}}}=0.$ \\
Now, we suppose that $n$ is odd. As it is mentioned in the previous proposition (\ref{hqc}), we just have to prove that $\displaystyle{\sum_{i=0}^{n-1}}\frac{\partial\{x_0, x_i\}}{\partial x_i}=0.$ Using the $\sigma$-invariance and antisymmetric condition, we have:
$$\displaystyle{\sum_{i=0}^{n-1}}\frac{\partial\{x_0, x_i\}}{\partial x_i}=\displaystyle{\sum_{i=1}^{\frac{n-1}{2}}}\left(\frac{\partial\{x_0, x_i\}}{\partial x_i}-\sigma^{n-i}\cdot(\frac{\partial\{x_0, x_i\}}{\partial x_0})\right).$$
Since $\{x_0, x_i\}=\displaystyle{\sum_{l=0}^{n-1}}p_l^ix_lx_{i-l},$ the term of this sum with $x_i$ is $(p_0^i+p_i^i)x_0x_i$ and the term $x_0$ is also $(p_0^i+p_i^i)x_0x_i.$ Hence:
$$\frac{\partial\{x_0, x_i\}}{\partial x_i}=(p_0^i+p_i^i)x_0$$
$$\frac{\partial\{x_0, x_i\}}{\partial x_0}=(p_0^i+p_i^i)x_i.$$
Therefore $\frac{\partial\{x_0, x_i\}}{\partial x_i}-\sigma^{n-i}\cdot(\frac{\partial\{x_0, x_i\}}{\partial x_0})=0,$ for all $i=1,\cdot, \frac{n-1}{2}$ and we obtain the result.
\end{proof}
\begin{rem}
One can note that in our proof, we just have used the antisymmetry and the $H$-invariance, but not the Jacobi identity. Then every quadratic element $Q\in\mathcal X^2(\mathcal A)$ which is $H$-invariant satisfies $D_2(Q)=0$.
\end{rem}
\begin{cor}
The Poisson algebras $q_{n, k}(\mathcal E)$ are unimodular.
\end{cor}
\subsection{$H$-invariant Poisson structure and Calabi-Yau algebras}
First let us remind the notion of Calabi-Yau (CY) algebras. These algebras plays the role of noncommutative
analogues of coordinate rings of CY manifolds.
\begin{definition} \cite{ginz}
An algebra $A$ is said to be a Calabi-Yau algebra of dimension $d$ if it is homologically smooth and
\begin{equation}
HH^k(A, A\otimes A)=\left\{ \begin{array}{l}
                      A \ \ \ if \ \ \ k=d \\
                      0 \ \ \ if \ \ \ k\neq d.
                   \end{array}
                   \right.
\end{equation}
\end{definition}
There are many ``polynomial-like'' algebras have been proved to be Calabi-Yau. Among them there are the polynomial rings $\mathbb{C}[x_1,\dots,x_n]$, the coordinate ring $\mathbb{C}[X]$ where $X$ is an affine smooth variety with
trivial canonical bundle, the Weyl algebra $\mathcal{A}_n=\mathbb{C}[\overline{x},\overline{\partial}]$,
where $\overline{x}=(x_1,\dots,x_n)$,
$\overline{\partial}=\left(\frac{\partial}{\partial x_1},\dots,\frac{\partial}{\partial x_n} \right)$, etc. .\\
It is well known that the Sklyanin algebras $Q_{3,1}(\mathcal E, \eta)$ and $Q_{4,1}(\mathcal E, \eta)$ are CY algebras since they are Kozsul algebras and their dual algebras are Frobenius algebras \cite{smsf,vanden}.\\
Recently, Dolgushev proved \cite{dol} some relation between CY algebras and the notion of unimodularity. Let $X$ be a smooth affine variety over $\mathbb C$ with a trivial canonical class and with a Poisson structure $\pi_1$ on its algebra of smooth functions. Following M. Kontsevich \cite{kon}, there is a bijection between the set of equivalence classes of star-products $[\star_h]$ of $\pi_1$ and the set of equivalence classes $[\pi_h]$ of Poisson tensors
$$\pi_h=0+\pi_1h+\cdots+\pi_nh^n+\cdots$$
on the commutative $\mathbb C[[h]]$-algebra $A[[h]].$\\
Dolgushev proved \cite{dol} that $(A[[h]], \star_h)$ is a CY algebra if and only if $\pi_h$ is unimodular Poisson structure. It is obvious to observe that the canonical quantization (those which the star-product equivalence class corresponds to the equivalence class of a Poisson tensor $\pi_h=0+\pi_1h$)  is a CY algebra if $\pi_1$ is unimodular Poisson structure.
\begin{cor}
Every canonical quantization of $H$-invariant quadratic Poisson structure is a CY algebra.
\end{cor}

In particular, every canonical quantization of an elliptic Sklyanin-Odesskii-Feigin Poisson algebra is a CY algebra.

\section{Classification of the H-invariant quadratic Poisson tensors until the dimension 6}
The generic entry of an $H$-invariant quadratic antisymmetric matrix is
\begin{equation}
 P_{ij}=\sum_{l=0}^{n-1}  p^{l}_{ij} x_{l}x_{i+j-l},
\end{equation}
with the condition $p_{i, j}^i+p_{i, j}^j=p_{i+1, j+1}^{i+1}+p_{i+1, j+1}^{j+1}.$ \\
We impose the Jacobi condition on the coefficients in lower dimensional cases.
\subsubsection*{dimension 3}
The first nontrivial case is $n=3$ when the generic structure of an $H$-invariant antisymmetric matrix is
\begin{equation}
\left(
\begin{array}{ccc}
 0       & P_1^0    & -P_1^2    \\
- P_1^0  &    0     & P_1^1   \\
  P_1^2 & - P_1^1  &    0
\end{array}
\right)
 \end{equation}
where $P_1^0=A_1x_0x_1+A_2x_{2}^2$. For every couple of coefficients $A_1,A_2$ the previous matrix respect the Jacobi
identities and they correspond exactly to the $q_{3,1}$ Poisson tensor related to the classical counterpart of the $Q_{3,1}$ Artin-Schelter -Tate algebras.
It is worst to remark that the solutions with $A_1 = 1, A_2 = 0$ and $A_1 = 0, A_2 =1$ correspond to the degenerations of the underlying elliptic
curve. The corresponding Poisson algebras are usually called "skew-polynomial" (A.Odesskii)
\subsubsection*{dimension 4}
 In dimension four, the matrix (\ref{H-P}) have the form:
\begin{equation}
\left(
\begin{array}{ccccc}
 0       & P_1^0    & P_2^0  & - P_1^3   \\
- P_1^0  &    0     & P_1^1  &  P_2^1   \\
  P_2^2 & - P_1^1  &    0   &  P_1^2     \\
 P_1^3 &  P_2^3     & -P_1^2 &0
\end{array}
\right)
 \end{equation}
where $P_1^0=A_1x_0x_1+A_2x_2x_3$ and $P_2^0=B_1x_0x_2+B_2x_{3}^2+B_3x_{1}^2$.
We see that in the even case the antisymmetry impose that $P_2^0=-P_2^2$ and $P_2^1=-P_2^3$ corresponding to the conditions
$$
B_1x_0x_2+B_2x_{3}^2+B_3x_{1}^2=-B_1x_2x_0-B_2x_{1}^2-B_3x_{3}^2
$$
and
$$
B_1x_1x_3+B_2x_{0}^2+B_3x_{2}^2=-B_1x_3x_1-B_2x_{2}^2-B_3x_{0}^2
$$
implying $B_1=0$, $B=-B_2=B_3$. Therefore $P_2^0=B(-x_{1}^2+x_{3}^2)$.\\
The Jacobi identities imply $A_1A_2=B^2$. In this case also the obtained Poisson tensors are the $q_{4,1}$ elliptic Poisson algebras with the normalization $A_1=1,A_2=k^2,B=k$.
\subsubsection*{dimension 5}
It is clear that Odesskii-Feigin Poisson algebras should verify the $H$-invariance condition. We observe in the previous subsections that the $H$-invariance quadratic Poisson structure coincides with Odesskii-Feigin algebras in dimension 3 and 4 or with their degenerations. We will show in this subsection that for the dimension $5$, apart of two elliptic Poisson structures, there exist also a cluster of skew-polynomial algebras and among them there $12$ $H$-invariant Poisson quadratic structures, which aris for some degeneration of the base curve in so-called "cuspidal points". But, of course, these "additional" structures do not extend the class of $H$-invariant quadratic Poisson algebras.

We will describe this degenerations in details in the paper \cite{ORT2}.

The generic form of an antisymmetric $H$-invariant quadratic matrix is:
\begin{equation}
\left(
\begin{array}{ccccc}
 0       & P_1^0    & P_2^0  & - P_2^3  & -P_1^4  \\
- P_1^0  &    0     & P_1^1  &  P_2^1  &  - P_2^4 \\
  -P_2^0 & - P_1^1  &    0   &  P_1^2  & P_2^2    \\
 P_2^3 &  -P_2^1     & -P_1^2 &0 &      P_1^3  \\
 P_1^4 & P_2^4 & -P_2^2 & -P_1^3 &0
\end{array}
\right)
 \end{equation}
where $$P_1^k=A_1 x_{k}x_{1+k} +A_2 x_{2+k}x_{4+k} +A_3 x_{{3+k}}^2$$ and
$$P_2^{k}=B_1 x_{k}x_{2+k}+B_2 x_{3+k}x_{4+k} +B_3 x_{{1+k}}^2.$$ The constants $A_i$ and $B_i$ should be determined by
the Jacobi identities. For a $5d$ tensor we have in general 10 independent equations from Jacobi, however the $H$-invariance reduce only to 2 independent, i.e. $\sum_{l \in \mathbb{Z}_5} P_{0l}\partial_lP_{12}+cyc(i,j,k)=0$ and
$\sum_{l \in \mathbb{Z}_5} P_{0l}\partial_lP_{13}+cyc(i,j,k)=0$. These PDE's can be reduced to conditions on the coefficients of cubic polynomial in $x_i$ which are equivalent to the system of quadratic equations.
\footnote{During an interesting discussion with Prof. Odesskii, we learned that he independently obtained the same result, but it is still unpublished.}
\begin{equation}
 \label{Jac5}
\begin{array}{rl}
{B_{{2}}}^{2}+3\,A_{{1}}A_{{3}}+B_{{1}}A_{{3}}+A_{{2}}B_{{3}}&=0\\
2\,{A_{{3}}}^{2}-2\,A_{{2}}B_{{1}}-A_{{1}}A_{{2}}+B_{{2}}B_{{3}}&=0\\
-{A_{{2}}}^{2}-3\,B_{{1}}B_{{3}}+A_{{1}}B_{{3}}+B_{{2}}A_{{3}}&=0\\
-2\,{B_{{3}}}^{2}-2\,B_{{2}}A_{{1}}+B_{{1}}B_{{2}}-A_{{2}}A_{{3}}&=0
\end{array}
\end{equation}
First of all the Poisson algebras $q_{5,1}$ and $q_{5,2}$ are solutions of this system with coefficients
\begin{equation}
\begin{array}{r}
 A_1= -\frac{3}{5}\lambda^2+\frac{1}{5\lambda^3} \quad A_2=-\frac{2}{\lambda} \quad A_3=\frac{1}{\lambda^2}\\ \\
 B_1= -\frac{1}{5}\lambda^2-\frac{3}{5\lambda^3} \quad B_2= 2 \quad B_3= \lambda
\end{array}
\end{equation}
for $q_{5,1}$ and
\begin{equation}
\begin{array}{r}
 A_1=\frac{2}{5}\lambda^2+\frac{1}{5\lambda^3} \quad A_2= \lambda \quad A_3=-\frac{1}{\lambda}\\ \\
 B_1=-\frac{1}{5}\lambda^2+\frac{2}{5\lambda^3} \quad B_2=- \frac{1}{\lambda^2} \quad B_3= 1
\end{array}
\end{equation}
for $q_{5,2}$.
In the forthcoming paper \cite{ORT2} we will study carefully this couple of Poisson tensors and we will show that they are related by a birational  homomorphism.

Apart of this two strata there is a cluster of solutions of skew-polynomial type which we can obtain $via$ direct computation using the algebraic manipulator Maple. It is not interesting, in our opinion, to write down explicitly all the classes but we analyze some interesting cases.
\begin{enumerate}
\item
An interesting solution is the linear solution
\begin{equation}
\begin{array}{r}
 A_1=-\frac{\lambda}{2}+1 \quad A_2=\lambda \quad A_3=-\frac{\lambda}{2}-1\\ \\
 B_1=\frac{\lambda}{2}+1\quad B_2=-2 \quad B_3=-\frac{\lambda}{2}+1
\end{array}
\end{equation}
This Poisson structure does not generates any integrable system because the Casimir of the structure is independent on $\lambda$.
\item We can extend the parameter $\lambda\in \mathbb C^*$ and consider it like a point in $\mathbb P^1: \lambda = (p:q)$.
The points $\lambda$ parametrises the base elliptic curve considering like a schematic intersection of five Klein quadrics in $\mathbb P^4$
(\cite{Hul}).
Consider the cases when $p=1, q=0$ and $p=0, q=1$. The corresponding skew polynomial Poisson brackets are associated with two pentagonal
configurations in $\mathbb P^4$ (degenerations of the schematic intersections of Klein quadrics).
\item There are also 10 similar "cuspidal" points in $\mathbb P^1$ (they are solutions of the equation $\lambda^10 + 11\lambda^5 - 1 = 0.)$
\end{enumerate}

\subsubsection*{Explicit form of $q_{5,k=1,2}$}
The explicit formulas for Sklyanin-Odesskii-Feigin Poisson brackets with $n=5$ can be extracted from
\cite{kiev} and \cite{odru}:

For the tensor $q_{5,1}$ the corresponding Poisson brackets are
\begin{equation}
 \begin{split}
 \{x_i,x_{i+1}\}_{5,1}&= \left(- \frac{3}{5}\,{\lambda}^{2}+\frac{1}{5{\lambda}^{3}} \right) x_{{i}}x_{{i+1}}-2\,{\frac {x_{{i+4}}x_{{i+2}}}{\lambda}}+{\frac {{x_{{i+3}}}^{2}}{{\lambda}^{2}}}\\
\{x_i,x_{i+2}\}_{5,1}&= \left( -\frac{1}{5}\,{\lambda}^{2}-\frac{3}{5{\lambda}^{3}} \right) x_{{i+2}}x_{{i}}+2\,x_{{i+3}}x_{{i+4}}-\lambda\,{x_{{i+1}}}^{2}
 \end{split}
\end{equation}
where $i\in \mathbb Z_5$.\\
For the tensor $q_{5,2}$ the bracket are
\begin{equation}
 \begin{split}
 \{y_i,y_{i+1}\}_{5,2}&= \left( \frac{2}{5}\,{\lambda}^{2}+\frac{1}{5{\lambda}^{3}} \right) y_{{i}}y_{{i+1}}+\lambda y_{{i+4}}y_{{i+2}}-{\frac {{y_{{i+3}}}^{2}}{{\lambda}}}\\
\{y_i,y_{i+2}\}_{5,2}&= \left( -\frac{1}{5}\,{\lambda}^{2}+\frac{2}{5{\lambda}^{3}} \right) y_{{i+2}}y_{{i}}-\frac{y_{{i+3}}y_{{i+4}}}{\lambda^2}+{y_{{i+1}}}^{2}
 \end{split}
\end{equation}
where $i\in \mathbb Z_5$.

In $5d$ the Poisson structure admits always only one Casimir. The $H$-invariance allows us to write it explicitly:
\begin{equation}
\label{Casgen5}
\begin{array}{rl}
K_5=&c_5 \left( {x_0}^{5}+{x_{{1}}}^{5}+{x_2}^{5}+{x_3}^{5
}+{x_4}^{5} \right)  \\
&+ c_4 \left( {x_0
}^{3}x_1x_3+{x_1}^{3}x_0x_2+{x_2}^{3}x_1x_
{{4}}+{x_3}^{3}x_2x_4+{x_4}^{3}x_0x_3 \right) \\
&+ c_3  \left( {x_0}^{3}x_2x_{{4}
}+{x_1}^{3}x_3x_4+{x_2}^{3}x_0x_4+{x_3}^{3
}x_1x_0+{x_4}^{3}x_2x_1 \right)  \\
&+  c_2 \left( x_0{x_1}^{2}{x_4}^{2}+x_1{x
_{{3}}}^{2}{x_0}^{2}+x_2{x_3}^{2}{x_1}^{2}+x_3{x_2}^{
2}{x_4}^{2}+x_4{x_0}^{2}{x_3}^{2} \right) \\
&+  c_1 \left( x_0{x_2}^{2}{x_{{4
}}}^{2}+x_1{x_3}^{2}{x_4}^{2}+x_2{x_4}^{2}{x_0}^{2}
+x_3{x_1}^{2}{x_0}^{2}+x_4{x_2}^{2}{x_1}^{2} \right)  \\
&+  c_0 x_0x_1 x_2x_3x_4
\end{array}
\end{equation}
where
\begin{equation}
\begin{array}{l}
 c_5= -\frac{1}{5} A_3 B_3 \quad  c_4= A_1 A_3 \quad  c_3= -B_1B_3 \quad
c_2= \frac{1}{2}A_1A_2 - \frac{1}{2}B_2 B_3 \\ \\ \quad  c_1=\frac{1}{2}A_2A_3 - \frac{1}{2}B_1 B_2  \quad
c_0= A_{1}^2 -B_{1}^2 -A_1B_1 -A_2B_2\\
\end{array}
\end{equation}
and the constants $A_i$ and $B_i$ satisfy (\ref{Jac5}).\\
The $H$-invariance implies also the following interesting property
\begin{equation}
\label{qr5}
 P_{ij}P_{kl}+P_{ki}P_{jl}+P_{jk}P_{il}=-2\,{\frac {B_{{2}}A_{{1}}+{B_{{3}}}^{2}}{A_{{2}}A_{{3}}-B_{{1}}B_{{2}
}}}\frac{\partial K_{5}}{\partial x_m},
\end{equation}
where $(i, j, k, l, m)$ is a cyclic permutation of $( 0, 1, 2, 3, 4 )$.\\
We could interpret this relation as a system of $5$ quadric for $6$ $P_{ij}$ entries for a fixed polynomial $K_5$.
We see, therefore, that contrarily to what happens in dimensions $3$ and $4$, that the Casimir do not encodes all the information necessary to the reconstruction of the Poisson tensor. This is related to the fact that the $5d$ is the smallest dimension for what the elliptic Poisson tensors are not of Jacobi type (see \cite{odru}).
\begin{rem}
 The relation (\ref{qr5}) in $5$ dimensions  is a generalization to $H$-invariant Poisson algebras of an analogue of the formula present in \cite{odru} for $q_{n,k}$. In this paper the authors consider this relation in their Theorem 3.1 stating that for any Poisson algebra of dimension $d$ with regular structure of symplectic leaves the following relation holds:
\begin{equation}
\label{ORthm}
 \Lambda^{\frac{d-l}{2}} \pi  = \star^{-1}\left( \Lambda_{i=1 \dots l} dQ_i \right).
\end{equation}
 Here $\pi$ is the Poisson structure and $\{Q_i\}_{i=1 \dots l}$ are the Casimir functions. This theorem could be useful for proving, in more geometric way, the unimodularity of the $q_{n,k}$ structures. Actually the formula (\ref{ORthm}) implies that  $\star^{-1} \circ d \circ \star \left( \Lambda^{\frac{d-l}{2}} \pi \right) = 0 $.
\end{rem}

\subsubsection*{dimension 6}
 The generic form of an antisymmetric $H$-invariant quadratic matrix is
\begin{equation}
\label{6Hinv}
\left(
\begin{array}{cccccc}
 0       & P_1^0    & P_2^0  & P_3^0 &- P_2^4  & -P_1^5  \\
- P_1^0  &    0     & P_1^1  &   P_2^1 & P_3^1 &  - P_2^5 \\
  -P_2^0 & - P_1^1  &    0   &  P_1^2  & P_2^2    & P_3^2\\
 P_3^3 &  -P_2^1     & -P_1^2 &0 &      P_1^3  &P_2^3\\
 P_2^4 & P_3^4 & -P_2^2 & -P_1^3 &0 & P_1^4 \\
P_1^5 & P_2^5 & P_3^5 & -P_1^3 & -P_1^4 &0
\end{array}
\right)
 \end{equation}
where
$$
\begin{array}{l}
P_1^k=A_1 x_{k}x_{1+k} +A_2 x_{2+k}x_{5+k} +A_3 x_{{3+k}}x_{{4+k}} \\
P_2^{k}=B_1 x_{k}x_{2+k}+B_2 x_{3+k}x_{5+k} +B_3 x_{{1+k}}^2+B_4 x_{{4+k}}^2 \\
P_3^k=C_1 x_{k}x_{3+k} +C_2 x_{4+k}x_{5+k} +C_3 x_{{2+k}}x_{{1+k}}.
\end{array}
$$
Because we are analyzing an even dimensional case the antisymmetry impose some constraint on $P_3^k$. From (\ref{6Hinv})
we see that we have to impose $P_3^k=-P_3^{k+3}$ equivalent to $C_1=0$ and $C=C_3=-C_2$. Therefore
$P_3^k=C(x_{{2+k}}x_{{1+k}}- x_{4+k}x_{5+k})$.\\
The cubic Jacobi relations are equivalent to the quadric relations on the coefficients
\begin{equation}
 \label{Jac6}
\begin{array}{rl}
{B_{{2}}}^{2}+CA_{{2}}-{A_{{3}}}^{2}&=0 \\CB_{{2}}-2\,B_{{3}}B_{{4
}}-A_{{2}}B_{{3}}&=0 \\A_{{2}}A_{{1}}-B_{{4}}A_{{3}}+B_{{1}}A_{{2}}-B_{{2}}
B_{{3}}&=0 \\CB_{{4}}-B_{{1}}A_{{3}}-A_{{1}}A_{{3}}&=0 \\CB_{{3}}+B_{{2}}B_{{1}}
+A_{{1}}B_{{2}}&=0 \\-2\,{B_{{3}}}^{2}+2\,CA_{{1}}-CB_{{1}}-B_{{4}}A_{{2}}&=0 \\
-B_{{2}}B_{{4}}-A_{{3}}B_{{3}}&=0 \\-{A_{{2}}}^{2}-2\,CB_{{1}}-2\,B_{{4}}A_
{{2}}-CA_{{3}}+CA_{{1}}&=0 \\-{C}^{2}-2\,B_{{1}}A_{{2}}+2\,B_{{4}}A_{{3}}-2
\,B_{{2}}B_{{3}}-A_{{2}}A_{{3}}+A_{{2}}A_{{1}}&=0 \\B_{{1}}A_{{2}}-B_{{4}}A
_{{3}}-B_{{4}}A_{{1}}&=0 \\B_{{2}}B_{{4}}-2\,B_{{3}}B_{{1}}+A_{{1}}B_{{3}}-
A_{{2}}B_{{2}}&=0
\end{array}
\end{equation}
The system (\ref{Jac6}) admits many classes of solutions
which can be computed using an algebraic manipulator (Maple in our case) similarly to  $5$ dimensional case. The algebra $q_{6,1}$ is obviously one of the solutions. Another interesting solution is obtained remarking that the $H_6$ Heisenberg group contains two copies of $H_3$. This implies that a particular solution of the system (\ref{Jac6}) is given by the direct sum of two copies of the $q_{3,1}$  Poisson algebra. This solution is $C=A_1=A_2=A_3=B_2=B_3=0$ and $B_1$, $B_4$ free coefficients. In this case the Poisson algebra is a direct sum of a copy of $q_{3,1}$ in the variables $x_0,x_2,x_4$ and a copy of $q_{3,1}$ in the variables $x_1,x_3,x_5$.

A. Odesskii (private communication) has conjectured that all the solutions of the system (\ref{Jac6}), except the previous examples, are some skew-polynomial degenerations of $q_{6,1}(\cal E)$.

\begin{rem}
The case of $d=6$ generators is related to some integrable systems.
In \cite{odru1} authors had constructed a simple model of a (quantum) elliptic integrable system in $Q_{2k,1}(\cal E), k\geq 2$. They had conjectured that this system is a quantization of the classical integrable system which arises
via the Magri-Lenard scheme from the bi-Hamiltonian structure on $q_{2k,1}(\cal E)$ described in \cite{ode3}.
Roughly speaking, the integrable system has as the phase
space a $2k$ dimensional component of the moduli space of parabolic rank two
bundles on the given elliptic curve $\cal E$, or, more precisely, the coordinate ring of
the open dense part of the latter has a structure of a quadratic Poisson algebra
isomorphic to $q_{2k,1}(\cal E)$. The quantum commuting
elements from the construction of \cite{odru1} are the same as the latter obtained from the
Lenard-Magri scheme for $k = 3$.
The algebra $q_{6,1}(\cal E)$ is the first non-trivial case when this conjecture was verified by a direct computation.
It would be interesting to find if there are some degenerations (defined on the solutions of (\ref{Jac6})) of this integrable system and to clarify their intrinsic meaning. Remind that the elliptic integrable system related to
$q_{6,1}(\cal E)$ was interpreted in \cite{odru1} as a Beauville-Mukai system on the (desingularized) part of
the symmetric product $(Cone(\cal E))^{(3)}$.
\end{rem}

\section{Conclusions}

In this paper we study the structure of the $H$-invariant polynomial Poisson tensors mainly in the quadratic case, and we classify them until the order $6$. We show that our construction generalize the Odesskii-Feigin Poisson tensors $q_{n,k}$ and that all the $H$-invariant quadratic bi-vectors are unimodular.

It is interesting to concentrate our attention and future research to the $5d$ case. There many different reasons to be interested in this particular case:
\begin{itemize}
\item It is the first case when there are two different Poisson tensors generated by the Odesskii-Feigin construction;

\item The family of quadrics in $\mathbb P^4$ which defines the underlying normal elliptic curve is only a \emph{local} complete intersections;

\item The bi-Hamiltonian properties of $q_{5,2}(\mathcal E)$  are still obscure while the algebra $q_{5,1}(\mathcal E)$ is in fact \emph{tri-Hamiltonian} as it was shown by A. Odesskii in \cite{ode2};

\item We are still do not know how to relate to these algebras (as well as to all other odd-dimensional algebras for $n>3$) an integrable system.

\item The projective geometry associated with with $n=5$ (together with the case $n=7$) cases will give probably good ideas and hints how to treat the general algebras with any \emph{prime} number $p$ of generators.

\end{itemize}

We will  display in the future paper \cite{ORT2} a particular (classical quadro-cubic) Cremona transformations in $\mathbb P^4$ between two different elliptic Sklyanin-Odesskii-Feigin Poisson structures with $5$ generators.\\
We will give a strong evidence for the existence of a similar Cremona description for morphisms between elliptic Sklyanin and Poisson algebras $q_{7,k}(\mathcal E),\ k = 1, 2, 3$ based on the results of \cite{HulKatzS}

We are going to study the structure and the properties of the integrable systems related to elliptic Poisson structures. It could be interesting to extend the algebraic construction of integrable systems on $\mathbb C^{2k}$ for $k\geq 3$ introduced in \cite{odru1} to odd-dimensional case and to relate it to the bi-Hamiltonian construction of  \cite{ode2}. Next we would like to understand the geometrical meaning of unimodularity for the integrable flows related to the elliptic bi-Hamiltonian Poisson structures.

\subsubsection*{Acknowledgements}
During the work the authors benefit of many useful discussions and suggestions, in particular with Marco Pedroni on the Poisson geometry and with Alexander Odesskii on the elliptic algebras.

G.O. thanks  MISGAM project (exchange grants no 1788 and 2760) for the financial support, the Mathematics Departments of the University of Angers for the kind hospitality, and Andrea Previtali and Stefano Meda for the some explanation about the discrete Heisenberg group. S.P. is thankful to University of Milan Bicocca and in particular to Gregorio Falqui for the invitation and the kind hospitality.  S.P. has been partially financed by ``Fond National de Recherche (Luxembourg)''. V.R. benefited a lot of numerous conversations with Boris Feigin, Alexei Gorodentsev and Sasha Polishchuk.

Part of this work was presented by V.R. on International Conferences
in Italy (Vietry, 06/09 and Trieste, 07/09), in Russia (Dubna, 08/09), Ukraine (Kiev-Odessa, 05/09) and in France (Lyon 11/09). His work and visits were partially supported by
MISGAM and "EINSTEIN" (RFBR-09-01-92440-KE-) programs, by PICS-4769 CNRS-Ukraine, French-Russian Project in Theoretical Physics (GDRI- 471), by RFBR grants 08-01-00667 and 09-02-90493-Ukr-f. He greatly acknowledges the hospitality of SISSA (Trieste) and Universities of Genova and Milan (Bicocca) where a part of this work was done. He acknowledges also a warm hospitality of LPTM (Universit\'e Cergy-Pontoise) where the paper was finished during his CNRS delegation.

\end{document}